\documentclass[10pt,oneside,final]{IEEEtran}
\usepackage{amssymb}
\usepackage{amsmath}

\usepackage{graphicx}
\usepackage{color}
\usepackage{multicol}

\usepackage{subfigure}
\usepackage{bm}
\usepackage{latexsym}
\usepackage{stfloats}
\usepackage{float}
\usepackage{exscale}
\usepackage{relsize}
\usepackage{cite}
\usepackage{cases}
\usepackage{algorithm}
\usepackage{algpseudocode}
\usepackage{amsthm}
\usepackage{epsfig}
\usepackage{epstopdf}
\usepackage{breqn}
\usepackage{enumerate}
\usepackage{multirow}
\usepackage[utf8]{inputenc}
\usepackage{algorithmicx}

\begin{document}
\title{\vspace{-1cm}Multi-Pair D2D Communications Aided by An Active RIS over Spatially Correlated Channels with Phase Noise}
\author{Zhangjie~Peng,
        Xue~Liu,
        Cunhua~Pan,~\IEEEmembership{Member,~IEEE},
        Li~Li, and
        Jiangzhou~Wang,~\IEEEmembership{Fellow,~IEEE}\vspace{-0.8cm}
\thanks{This work was supported in part by the Natural Science Foundation of Shanghai under Grant 22ZR1445600, in part by the open research fund of National Mobile Communications Research Laboratory, Southeast University under Grant 2018D14, and in part by the National Natural Science Foundation of China under Grant 61701307. \emph{(Corresponding authors: Cunhua Pan and Xue Liu.)}}
\thanks{Zhangjie Peng is with the College of Information, Mechanical and Electrical Engineering, Shanghai Normal University, Shanghai 200234, China, also with the National Mobile Communications Research Laboratory, Southeast University, Nanjing 210096, China, and also with the Shanghai Engineering Research Center of Intelligent Education and Bigdata, Shanghai Normal University, Shanghai 200234, China (e-mail: pengzhangjie@shnu.edu.cn).}
\thanks{X. Liu and L. Li are with the College of Information, Mechanical and Electrical Engineering, Shanghai Normal University, Shanghai 200234, China (e-mail: 1000494962@smail.shnu.edu.cn; lilyxuan@shnu.edu.cn).}
\thanks{C. Pan is with the National Mobile Communications Research Laboratory,
Southeast University, Nanjing 210096, China. (e-mail: cpan@seu.edu.cn).}
\thanks{J. Wang is with the School of Engineering, University of Kent, CT2 7NT Canterbury, U.K. (e-mail: j.z.wang@kent.ac.uk).}
}
\maketitle
\newtheorem{lemma}{Lemma}
\newtheorem{theorem}{Theorem}
\newtheorem{remark}{Remark}
\newtheorem{corollary}{Corollary}
\newtheorem{proposition}{Proposition}\vspace{-0.3cm}
\begin{abstract}
This paper investigates a multi-pair device-to-device (D2D) communication system aided by an active reconfigurable intelligent surface (RIS) with phase noise and direct link. The approximate closed-form expression of the ergodic sum rate is derived over spatially correlated Rician fading channels with statistical channel state information (CSI). When the Rician factors go to infinity, the asymptotic expressions of the ergodic sum rates are presented to give insights in poor scattering environment. The power scaling law for the special case of a single D2D pair is presented without phase noise under uncorrelated Rician fading condition. Then, to solve the ergodic sum rate maximization problem, a method based on genetic algorithm (GA) is proposed for joint power control and discrete phase shifts optimization. Simulation results verify the accuracy of our derivations, and also show that the active RIS outperforms the passive RIS.
\end{abstract}
\begin{IEEEkeywords}
Reconfigurable intelligent surface (RIS), active RIS, device-to-device (D2D) communication, ergodic sum rate, phase noise, spatial correlation, power control.\vspace{-0.4cm}
\end{IEEEkeywords}
\section{Introduction}
Recently, reconfigurable intelligent surface (RIS) has emerged as a brand new communication paradigm to reconfigure the radio propagation environment in a desired manner \cite{1}.
To be specific, RIS is an array of reflecting elements, each of which can independently induce a phase shift on the incident signals. By carefully tuning the phase shifts, RIS can be utilized for reducing the sum power \cite{2} and enhancing the system sum rate \cite{3}. The RIS also possesses the advantages of small size, light weight, convenient deployment and low cost.

On the other hand, device-to-device (D2D) technology is also regarded as a promising solution to alleviate the capacity demand of local transmission. The integration of RIS in D2D communications has sparkled extensive research efforts \cite{4,5}. In \cite{4}, the power allocation scheme was presented and an RIS was deployed to improve the system sum rate by eliminating the interference of D2D communications. In \cite{5}, the energy efficiency was maximized through jointly optimizing the power allocation and the phase shifts in D2D communication network.

However, the RIS considered in \cite{1,2,3,4,5} is passive, which can only reflect the incident signal without any signal processing operations. The performance gain brought by the passive RIS is limited due to the “multiplicative fading” effect with strong direct links between the base station and the users \cite{6}. To address this issue, active RIS has been proposed, which can amplify the incident signals in the electromagnetic level and adjust the phase shifts, concurrently \cite{15}. Different from amplify-and-forward relays, the active RIS operates in full-duplex mode and is equipped with low-power reflection-type amplifiers instead of power-hungry radio frequency chains \cite{7}. The active RIS inherits the advantages of the passive RIS, while its superiorities have been demonstrated in \cite{8}. However, the contributions in \cite{6,15,7,8} were based on the assumption of the availability of instantaneous channel state information (CSI), which results in high channel estimation overhead. To tackle this issue, the authors in \cite{9,10,11,12} designed the RIS phase shifts exploiting statistical CSI, which varies much slower and can relax the necessity for configuring the RIS frequently. Besides, the impact of spatial correlation is non-negligible since the reflecting elements are located close to each other due to the small size of the RIS. The performance of RIS-aided cell-free massive multiple-input multiple-output (MIMO) system was analyzed under the presence of spatial correlation \cite{9}. The ergodic sum rate was studied and maximized in the RIS-aided MIMO multiple-access channel system over spatially correlated Rician fading with statistical CSI \cite{10}. The outage probability in RIS-aided communication system was investigated in \cite{11} and the impact of the Von Mises phase errors of the reflecting elements was further studied in \cite{12}.

In this paper, we investigate a multi-pair D2D communication system aided by an active RIS over spatially correlated Rician fading channels, where only statistical CSI is available to reduce the channel estimation overhead. It is noted that our work differs from \cite{18} where the RIS was used as a receiver. The contributions of this paper are summarized as follows: 1) The approximate closed-form expression of the ergodic sum rate is derived; 2) We propose a method based on genetic algorithm (GA) to maximize the ergodic sum rate through joint power control and discrete phase shifts optimization; 3) Simulation results validate the accuracy of our derivations and the superiority of the active RIS over the passive RIS.\vspace{-0.2cm}
\section{System Model}\vspace{-0.1cm}
Consider a D2D overlaying communication system as depicted in Fig. \ref{fig1}, which serves $K$ pairs of single antenna users. The base station acts as a controller, which not only allocates the power and the spectrum resources for the D2D user-pairs, collects all the statistical CSI estimated by the devices but also performs the power control and the optimization of RIS phase shifts. The active RIS with $N=N_HN_V$ reflecting elements is deployed to further enhance the communication performance. The phase shift matrix is expressed as $\bm{\Theta}={\rm{diag}}(e^{j\theta_1},\ldots,e^{j\theta_n},\ldots,e^{j\theta_N})\in {{\mathbb C}^{N \times N}}$, where $n=1,\cdots,N$. The phase noise is considered due to the hardware limit of the active RIS. The phase noise matrix is denoted as $\bm{\Phi}={\rm{diag}}(e^{j\tilde \theta_1},\ldots,e^{j\tilde \theta_n},\ldots,e^{j\tilde \theta_N})$ and $\tilde\theta_n$ follows the circular normal distribution with zero mean and concentration parameter $\kappa_{\tilde \theta}$ \cite{13}. The amplification factor matrix is ${\bm{\Lambda }} = {\rm{diag}}\left( { {{\eta _1}} ,\ldots,\eta_n,\ldots,{{\eta _N}} } \right)$. The amplification noise is ${{\bf{n}}_F}\sim{\cal C}{\cal N}\left( {{{\bf{0}}},\sigma _F^2{{\bf{I}}}} \right)$ and ${{\bf{n}}_F}\in {{\mathbb C}^{N \times 1}}$.
\begin{figure}[t]\vspace{-0.8cm}
\centering
\includegraphics[scale=0.25]{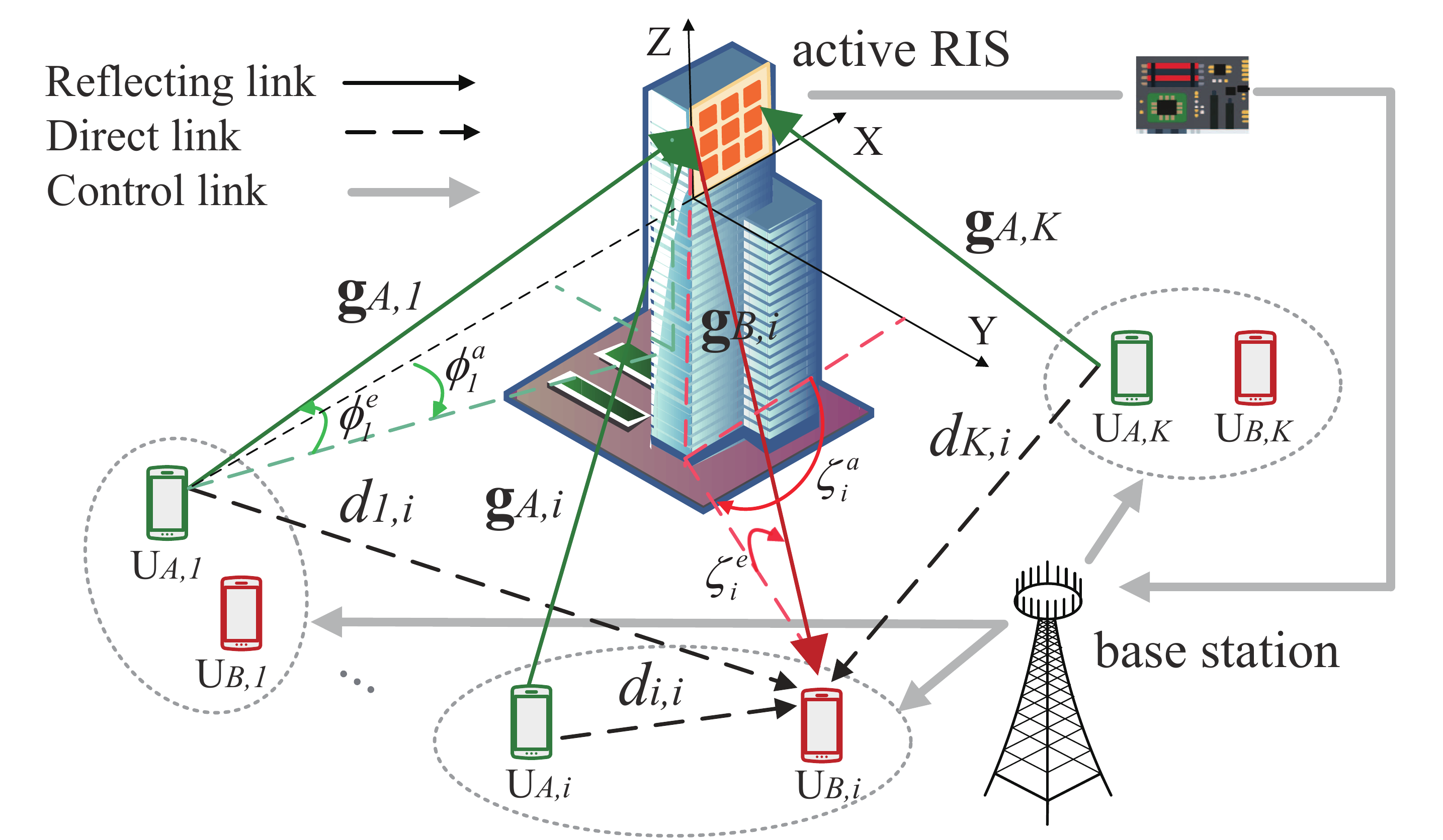}\vspace{-0.4cm}
\caption{System model of multi-pair D2D communication aided by an active RIS.} \label{fig1}\vspace{-0.7cm}
\end{figure}

We denote the $i$-th D2D user-pair as $\text U_{A,i}$ and $\text U_{B,i}$ for $i = 1, \cdots, K$. The transmit power of $\text U_{A,i}$ is $P_i$, and ${s}_{A,i}$ stands for the information symbol with unit power.

The direct channel link between $\text U_{A,i}$ and $\text U_{B,j}$ is denoted as $h_{i,j}$, which is assumed to follow the Rician fading in the short communication distance, i.e., $h_{i,j}\sim{\cal C}{\cal N}\left({{\bar h}_{i,j}},{\sigma_{i,j}^2}/({1+\gamma_{i,j}}) \right)$ where ${\bar h}_{i,j}={\sigma_{i,j}}\sqrt{\gamma_{i,j}/({1+\gamma_{i,j}})}$. ${{\bar h}_{i,j}}$ is the line-of-sight (LoS) component, $\sigma_{i,j}^2$ is the large-scale fading coefficient and $\gamma_{i,j}$ is the Rician factor.
The reflecting channels $\text U_{A,i}\to$ RIS and RIS $\to \text U_{B,i}$ are assumed to follow the spatially correlated Rician fading since the active RIS is usually deployed in high altitude with LoS components \cite{14}:
\setlength\abovedisplayskip{2pt}
\setlength\belowdisplayskip{2pt}
\begin{align}\label{gtiimp}
{{\bf{g}}_{A,i}}&  \!=\!\sqrt{{\alpha _i\gamma _{A,i}}/{(1+\gamma _{A,i})}}{{\bf{\bar g}}_{A,i}} \!+\!\sqrt{{\alpha _i}/{(1+\gamma _{A,i})}} {{\bf{\tilde g}}_{A,i}},\\
{{\bf{g}}_{B,i}}& \! =\!\sqrt{{\beta _i\gamma _{B,i}}/{(1+\gamma _{B,i})}}{{\bf{\bar g}}_{B,i}} \!+\!\sqrt{{\beta _i}/{(1+\gamma _{B,i})}} {{\bf{\tilde g}}_{B,i}},
\end{align}
where ${{\bf{g}}_{A,i}}$, ${{\bf{g}}_{B,i}}\in {{\mathbb C}^{N \times 1}}$. $\alpha_{i}$ and $\beta_{i}$ denote the large-scale fading coefficients. $\gamma_{A,i}$ and $\gamma_{B,i}$ are the Rician factors. ${{\bf{\tilde g}}_{A,i}}$ and ${{\bf{\tilde g}}_{B,i}}$ are non-line-of-sight (NLoS) components, which are distributed as ${{\bf{\tilde g}}_{A,i}}\sim{\cal C}{\cal N}\left( {\bf{0}},{\bf{R}} \right)$ and ${{\bf{\tilde g}}_{B,i}}\sim{\cal C}{\cal N}\left( {\bf{0}},{\bf{R}} \right)$, respectively. $\bf{R}$ stands for the normalized spatial correlation matrix under the isotropic scattering model, whose ($p,q$)-th element $[{\bf{R}}]_{p,q}$ is \cite[Eq. (10)]{14}.
\begin{align}\label{R}\hspace{-0.2cm}
r_{p,q}\!=\!{\rm{sinc}}\!\!\left(\!\!{{2{\sqrt {\!{{\left( {h\left( p \right)\!-\! h\left( q \right)} \right)}^2}d_H^2 \!+ \! {{\left({v\left( p \right)\!-\! v\left( q \right)} \right)}^2}d_V^2}}}\!\big/\!{\lambda }} \!\right)\!,
\end{align}
where ${h\left( p \right)}=\text{mod}\left( p-1,N_H \right)$ and ${v\left( p\right)}=\left\lfloor {\left( p-1\right)/N_H} \right\rfloor$. $d_H$ and $d_V$ are the horizontal and vertical element-spacing, $\lambda$ is the wavelength and we set $d_H=d_V=\frac{1}{4}\lambda$.
${{{\bf{\bar g}}}_{A,i}}$ and $ {{{\bf{\bar g}}}_{B,i}}$ represent the LoS components under the uniform planar array (UPA) model, which are respectively written as
\begin{align}
{{{\bf{\bar g}}}_{A,i}} &\!=\!\!{\left[\!\begin{array}{l}\!\!
1, \cdots, {e^{j\frac{{2\pi }}{\lambda }\left( {{h\left( n \right)}{d_H}\!\sin \phi _i^a{\rm{cos}}\phi _i^e + {v\left( n \right)}{d_V}\!\sin \phi _i^e} \right)}}, \ldots ,\!\!\\
\ {e^{j\frac{{2\pi }}{\lambda }\left( {\left( {{N_H} - 1} \right){d_H}\!\sin \phi _i^a\cos \phi _i^e + \left( {{N_V} - 1} \right){d_V}\!\sin \phi _i^e} \right)}}
\end{array} \!\!\right]^T}\!\!\!\!\!,
\end{align}\vspace{-0.1cm}
\begin{align}
{{{\bf{\bar g}}}_{B,i}} &\!= \!\!{\left[\!\begin{array}{l}\!\!
1, \ldots ,{e^{j\frac{{2\pi }}{\lambda }\left( {{h\left( n \right)}{d_H}\!\sin \zeta _i^a\cos \zeta _i^e + {v\left( n \right)}{d_V}\!\sin \zeta _i^e} \right)}}, \ldots ,\!\!\\
\ {e^{j\frac{{2\pi }}{\lambda }\left( {\left( {{N_H} - 1} \right){d_H}\!\sin \zeta _i^a\cos \zeta _i^e + \left( {{N_V} - 1} \right){d_V}\!\sin \zeta _i^e} \right)}}
\end{array}\!\! \right]^T}\!\!\!\!\!,
\end{align}
where $\phi _i^a$ and $\phi _i^e$ denote the azimuth and elevation angles of arrival (AoA), $\zeta _i^a$ and $\zeta _i^e$ denote the azimuth and elevation angles of departure (AoD).
The amplification power of the active RIS is given by
\begin{equation}\label{P_RIS}
\sum\nolimits_{i = 1}^K {{P_i}{\mathbb{E}}\left\{ {{{\left\| {{\bm{\Lambda\Theta\Phi }}{{\bf{g}}_{A,i}}} \right\|}^2}} \right\}}  + {\mathbb{E}}\left\{ {{{\left\| {{\bm{\Lambda\Theta\Phi }}}{{\bf{n}}_F} \right\|}^2}} \right\}= {P_R},
\end{equation}
where the expectation is taken over the NLoS components in the channels. The total power consumption of the active RIS aided system is
\begin{equation}\label{P_active}
{P_T} =\sum\nolimits_{i = 1}^K \!\!{{P_i}}+\varepsilon ^{ - 1}{P_R}+ N({P_{DC}} + {P_{SW}}),
\end{equation}
where $\varepsilon$ is the amplifier efficiency, $P_{DC}$ and $P_{SW}$ represent the direct current biasing power consumption and the power consumption of the switch and control circuit at each reflecting element \cite{15}.
In addition, the signal received at the $j$-th D2D receiver is given by
\begin{equation}\label{ybj}
{y_j} \!\!=\!\!\underbrace{\sqrt {\!{P_{j}}}{g_{j,j}}{s_{A,j}}}_{\text{desired signal}}\!+\!\!\!\underbrace{\sum\nolimits_{i \ne j}^K\!\!\!{\sqrt{\!{P_{i}}}{g_{i,j}}{s_{A,i}}}}_{\text{inter-pair interference}}\!+\underbrace{{\bf{g}}_{B,j}^T{\bm{\Lambda \Theta\Phi}} {{\bf{n}}_F}}_{\text{dynamic noise}} + \!\! \!\!\!\underbrace{{n_j}}_{\text{static noise}}\!\!\!,
\end{equation}
where $g_{i,j}={\bf{g}}_{B,j}^T{\bm{\Lambda \Theta\Phi}}{{\bf{g}}_{A,i}}+h_{i,j}$ and ${n_j}\sim{\cal C}{\cal N}\left( {0,\sigma _{j}^2} \right)$.

The signal-to-interference plus noise ratio (SINR) at $\text U_{B,j}$ is\vspace{-0.1cm}
\begin{equation}\label{SINR_j}
{\gamma _j} = \frac{{{P_{j}}{{\left| g_{j,j} \right|}^2}}}{{
\sum\nolimits_{i \ne j}^K \!{{P_{i}}{{\left| g_{i,j} \right|}^2}}
+ {{\left\| {{\bf{g}}_{B,j}^T\bm{\Lambda \Theta\Phi }}\right\|}^2}\sigma _F^2 + \sigma _j^2}}.\vspace{-0.1cm}
\end{equation}

Then, the ergodic rate is expressed as
\begin{align}\label{monte}
R_j= \mathbb{E}\big\{\text{log}_2(1+{\gamma _j})\big\},\vspace{-0.1cm}
\end{align}
and the ergodic sum rate is $R=\sum_{j=1}^K R_j$.\vspace{-0.4cm}
\section{Ergodic Sum Rate Analysis}
\begin{figure*}[b!]\vspace{-0.4cm}\hrulefill
\begin{align}\setcounter{equation}{15}\label{R_bar}
&{{\bar R}_j} = {\log _2}\left( {1 + \frac{{{P_j}\left( {{\alpha _j}{\beta _j}{\eta ^2}\left( {N+{\Gamma _{j,j}}} \right) + 2\eta \kappa {\sigma _{j,j}}\sqrt {{\alpha _j}{\beta _j}} {\Upsilon _{j,j}} + \sigma _{j,j}^2} \right)}}{{\sum\nolimits_{i \ne j}^K {{P_i}\left( {{\alpha _i}{\beta _j}{\eta ^2}\left( {N+{\Gamma _{i,j}}} \right) + 2\eta \kappa {\sigma _{i,j}}\sqrt {{\alpha _i}{\beta _j}} {\Upsilon _{i,j}} + \sigma _{i,j}^2} \right)}  + {\eta ^2}N{\beta _j}\sigma _F^2 + \sigma _j^2}}} \right),\\
\label{R_bar_pas}
&{{\bar R}_j^{pas}} = {\log _2}\left( {1 + \frac{{{P_j}\left( {{\alpha _j}{\beta _j}\left( {N+{\Gamma _{j,j}}} \right) + 2\kappa {\sigma _{j,j}}\sqrt {{\alpha _j}{\beta _j}} {\Upsilon _{j,j}} + \sigma _{j,j}^2} \right)}}{{\sum\nolimits_{i \ne j}^K {{P_i}\left( {{\alpha _i}{\beta _j}\left( {N+{\Gamma _{i,j}}} \right) + 2\kappa {\sigma _{i,j}}\sqrt {{\alpha _i}{\beta _j}} {\Upsilon _{i,j}} + \sigma _{i,j}^2} \right)}+ \sigma _j^2}}} \right).
\end{align}\vspace{-0.3cm}
\end{figure*}

In this section, the approximate closed-form expression of the ergodic sum rate is derived with statistical CSI, which contains the spatial correlation, the location and angle information of the transceivers. First, we assume that the amplification factor of each element is the same, i.e., ${{\eta _n}}= {{\eta}}$, which is derived below.\vspace{-0.15cm}
\begin{lemma}\label{lemma_eta}
The amplification factor for each element on the active RIS is given by\vspace{-0.1cm}
\begin{equation}\setcounter{equation}{12}\tag{11}\label{eta}
\eta  = {{{{\sqrt{P_R}}}}\big/{\sqrt{N\big( {\sum\nolimits_{i=1}^K {{P_i}{\alpha _i}}  + \sigma _F^2}\big )}}}.\vspace{-0.1cm}
\end{equation}
\end{lemma}
\begin{proof}
With ${\mathbb{E}}\left\{ {{{\bm{\Phi}}^H}{{\bf{\Theta }}^H}{{\bm{\Lambda }}^H}{\bm{\Lambda \Theta\Phi }}} \right\}
\!=\!\eta^2{{\bf{I}}_{N\times N}}$, we can derive ${\mathbb{E}}\{ {{{\left\| {{\bm{\Lambda\Theta\Phi }}{{\bf{g}}_{A,i}}} \right\|}^2}} \}\!=\!\frac{{{\alpha _i}}\eta^2}{{1 + {\gamma _{A,i}}}}(\gamma_{A,i}{\mathbb{E}}\left\{ {{\bf{\bar g}}_{A,i}^H{{{\bf{\bar g}}}_{A,i}}} \right\}+{\mathbb{E}}\left\{ {{\bf{\tilde g}}_{A,i}^H{{{\bf{\tilde g}}}_{A,i}}} \right\})=\frac{{{\alpha _i}}\eta^2}{{1 + {\gamma _{A,i}}}}(\gamma_{A,i}N+N)=\eta^2 N{\alpha_i}.$ Similarly, we have ${\mathbb{E}}\{ {{{\left\| {{\bm{\Lambda\Theta\Phi }}}{{\bf{n}}_F} \right\|}^2}} \}=\eta^2 N{\sigma _F^2}$. Substituting these results into \eqref{P_RIS} yields $\eta^2 N ( {\sum\nolimits_{i = 1}^K {{P_i}{\alpha _i}}  + \sigma _F^2})= {P_R}$.
\end{proof}
Based on Lemma \ref{lemma_eta}, the approximation of the ergodic rate is given in the following theorem.\vspace{-0.2cm}
\begin{theorem}\label{thm_appro}
For the considered multi-pair D2D communication system assisted by an active RIS over spatially correlated channels with phase noise, the closed-form expression of the approximate ergodic rate of $\text U_{B,j}$ is given by $R_j\approx{\tilde R_j}$, where
\begin{equation}\label{R_appro_act}
{\tilde R_j}\!\!=\!\!{\log _2}\!\!\left(\!\!{1 \!+\! \frac{{{P_j}\!\left({\eta^2\Omega _{j,j}}\!+\!2\eta c_{j,j}\Upsilon _{j,j}\!+\!\sigma_{j,j}^2\right)}}{\!\!\!\!{\sum\limits_{i \ne j}^K \!{{P_i}\!\left(\eta^2{\Omega _{i,j}}\!\!+\!2\eta c_{i,j}\Upsilon _{i,j}\!\!+\!\!\sigma_{i,j}^2\right)} \!\! +\!\!\eta^2N{\beta _j}\sigma _F^2  \!+\!\sigma _j^2}}} \!\!\right)\!\!,
\end{equation}
in which $\Omega _{i,j}\!= \!{\alpha _i}{\beta _j}N \!\!+\!{\tau _{i,j}}( {\gamma _{A,i}}{\gamma _{B,j}}{\Gamma _{i,j}} $$+{\gamma _{B,j}}L_{\zeta _i^a,\zeta _i^e}\! +\! {\gamma _{A,i}}L_{\phi _i^a,\phi _i^e}\! +\! {L_0}),\tau_{i,j}\!\!=\!\frac{\alpha _i\beta _j}{\left( 1+\gamma _{A,i} \right) \left( 1+\gamma _{B,j} \right)}$. $\Gamma _{i,j}$, $L_{x,y}$, $L_0$ and ${\Upsilon _{i,j}}$ are given in \eqref{Gamma}, \eqref{L}, \eqref{4} and \eqref{Upsilon} in Appendix \ref{Theorem1}.
\end{theorem}\vspace{-0.15cm}
\begin{proof}
See Appendix \ref{Theorem1}.
\end{proof}\vspace{-0.2cm}
\begin{remark}\label{remark1}
1) By setting the values of $\eta$ and $\sigma_F^2$ in \eqref{R_appro_act} with $\eta=1$ and $\sigma_F^2=0$, the ergodic rate in the passive RIS case is written as
\begin{equation}\label{R_appro_pas}
\tilde{R}_j^{pas}\! =\! {\log _2}\!\left(\!\! {1 \!+\! \frac{P_j^'\left({\Omega _{j,j}}\!+\!2c_{j,j}\Upsilon _{j,j}\!+\!\sigma_{j,j}^2\right)}{{\sum\nolimits_{i \ne j}^K \!P_i^'\!\!\left({\Omega _{i,j}}\!+\!2c_{i,j}\Upsilon _{i,j}\!+\!\sigma_{i,j}^2\right)\!+\! {\sigma _j^2}}} } \right)\!\!.\vspace{-0.05cm}
\end{equation}

2) Furthermore, when $\Omega _{i,j}=\Omega _{j,j}=0$ and $\Upsilon _{i,j}=\Upsilon _{j,j}=0$ for any $i$ and $j$, we derive the ergodic rate in the case without RIS, which is given by\vspace{-0.1cm}
\begin{equation}\label{R_appro_noris}
\tilde{R}_j^{no RIS} = {\log _2}\left( {1 + \frac{P_j^{''}{\sigma _{j,j}^2}}{{\sum\nolimits_{i \ne j}^KP_i^{''}{\sigma _{i,j}^2}+ {\sigma _j^2}}} } \right).\vspace{-0.1cm}
\end{equation}
\end{remark}

The total power consumption in the passive RIS case and the case without RIS are ${P_T^{pas}} =\sum\nolimits_{i = 1}^K{{P_i^{'}}}+ N{P_{SW}}$ and ${P_T^{noRIS}} =\sum\nolimits_{i = 1}^K{{P_i^{''}}}$, respectively. For fairness, we assume ${P_T}={P_T^{pas}}={P_T^{noRIS}}=P$. 
Besides, when the RIS hardware is ideal, i.e., ${\bf{\Phi}}={\bf {I}}_{N\times N}$, the ergodic rates without phase noise are rewritten by substituting $\kappa=1$ into \eqref{R_appro_act} and \eqref{R_appro_pas}.\vspace{-0.2cm}
\begin{corollary}\label{Rician_infinity}
When the Rician factors go to infinity, i.e., $\gamma_{A,i}=\gamma_{B,i}=\gamma_{A,j}=\gamma_{B,j} =\gamma_{i,j}=\gamma_{j,j}\to \infty$ for any $i$ and $j$, the ergodic rates ${\tilde R_j}$ in \eqref{R_appro_act} and ${\tilde R_j^{pas}}$ in \eqref{R_appro_pas} converge to
\begin{equation}\label{R_Rician}
{\tilde R_j} \to {\bar R_j}, {\tilde R_j^{pas}} \to {\bar R_j^{pas}},
\end{equation}
where ${\bar R_j}$ and ${\bar R_j^{pas}}$ are given by \eqref{R_bar} and \eqref{R_bar_pas}, respectively.
\end{corollary}\vspace{-0.1cm}
Corollary \ref{Rician_infinity} means that when the communication environment has limited scatters, the ergodic rates converge to fixed values depending on the phase shifts.\vspace{-0.2cm}
\begin{corollary}\label{coro_P_inf}
If the total power consumption $P\to\infty$ (i.e., $P_R\to\infty$, $P_i\approx P_i^'\approx P_j\approx P_j^'\to\infty$ for any $i$ and $j$) and the direct links are non-exist due to the blockage, we can derive that ${\tilde R_j} \to \dot R_j$ and ${\tilde R_j^{pas}} \to \dot R_j$, where\vspace{-0.1cm}
\begin{equation}\label{P_infty}\setcounter{equation}{18}
\dot R_j= {\log _2}\left( {1 + \frac{{{\Omega _{j,j}}}}{{\sum\nolimits_{i \ne j}^K {{\Omega _{i,j}}}}}} \right).\vspace{-0.1cm}
\end{equation}
\end{corollary}
Corollary \ref{coro_P_inf} shows that the ergodic rate of the active RIS aided system with high transmit power converges to the same value as the passive RIS counterpart. This is because the dynamic noise caused by active RIS and the static noise are negligible, and the power of both the desired signal and the multi-pair interferences increase by the same factor of $\eta^2$. In other words, the ergodic sum rate in passive RIS case is close to the active counterpart at high power consumption.\vspace{-0.2cm}
\begin{remark}\label{onepair}
If ${\bf{R}}={\bf{I}}_{N \times N}$, the reflecting links are simplified to
uncorrelated Rician fading channel. Then, we consider a special case that there is only one user-pair ($U_{A,i}$ and $U_{B,i}$) and the direct link is blocked, the ergodic rate becomes
$\tilde {{R_1}} = {\log _2}\!\big(\! {1 + \frac{{{P_i}{P_R}\Omega _{i,i}^*}}{{N{P_R}{\beta _i}\sigma _F^2 + N\left( {{P_i}{\alpha _i} + \sigma _F^2} \right)\sigma _i^2}}} \big)$ where $\Omega _{i,i}^*\!= \!{\alpha _i}{\beta _i}N \!\!+\!{\tau _{i,i}}{\gamma _{A,i}}{\gamma _{B,i}}{\Gamma _{i,i}}$.
Assuming that the hardware of the active RIS is ideal, if the transmit power $P_i$ is scaled down by $\frac{1}{{{N}}}$, i.e. ${P_i} = \frac{{{E_u}}}{{{N}}}$ for a fixed $E_u$, we have
$\tilde{R_1} \!\to\! {\log _2}\!\left( {1 \!+\!\frac{{{E_u}{P_R}{\tau _{i,i}}{\gamma _{A,i}}{\gamma _{B,i}}}}{{{P_R}{\beta _i}{\sigma _F}^2 +{ \sigma _F^2} \sigma _i^2}}}\! \right)\!\text{as } N \!\to\! \infty$.
\end{remark}\vspace{-0.05cm}
\begin{proof}
See Appendix \ref{N_infty}.
\end{proof}\vspace{-0.3cm}
\section{Power Control and Phase Shifts Optimization}
In this section, we propose a GA-based method to jointly optimize the transmit power and the phase shifts to maximize the ergodic sum rate. Due to the hardware limit, the discrete phase shift design is considered and the optimization problem is formulated as follows:
\setcounter{equation}{18}
\begin{subequations}\label{Optimization_Problem}
\begin{align}\hspace{-1.5cm}
\mathop {\max }\limits_{{\bm{p}}, {\bm{\theta}}}&\ \tilde R=\sum\nolimits_{j = 1}^K \tilde R_j\\
\label{s1}\text{s.t.}&\ 0< p_i\le p_{max},\forall i=1,\cdots,K,\\
\label{s2}&{\ \theta _n} \!\in\! \left\{\! {0,{{2\pi }}\!/{{{2^B}}}\!,...,\left(\! {{2^B}\!\! -\! 1} \!\right)\!{{2\pi }}\!/{{{{\rm{2}}^B}}}}\right\}\!,\forall n\!=\!1,\cdots,N,
\end{align}
\end{subequations}
where ${\bm{p}}=[p_1,p_2,\cdots,p_K]$, ${\bm{\theta}}=[\theta_1,\theta_2,\cdots,\theta_N]$ and $B$ is the number of bits for controlling the phase shifts.

Note that GA is a globally search method inspired by imitating the biological evolution, which can automatically accumulate knowledge about the search space, and adaptively control the search process to obtain the optimal solution. The details of our proposed GA-based method are given in Algorithm \ref{GA}.
\vspace{-0.3cm}
\begin{algorithm}[t] 
	\caption{GA-based method} 
	\label{GA} 
	\textbf{Input}: the population size $N_{po}$, the number of parents $N_{pa}$, the number of mutated chromosomes $N_{mu}$, the maximum iteration number $N_i$ and the termination fitness value $f$.
	
	\textbf{Output}: the elite ${\bm{e}}_t=[{\bm{p}_t},{\bm{\theta}_t}]$ of each generation.
	
	\textbf{Initialize}: the first generation $\bm{e}_{s}$ for $s=1,\cdots,N_{po}$,
	\begin{algorithmic}[1]
		\For{$t=1$ to $N_{i}$}
		\State {Calculate the fitness value $f(\bm{e}_{s})={\frac{1}{\tilde R({\bm{p}_{s}},{\bm{\theta}_{s}})}}$;}
		\State Note down $\bm{e}_{t}$ with the minimum fitness value;
		\If {$f({\bm{e}_{t}})>f$}
		\State Applying roulette wheel selection on the individ-
        \Statex\ \ \ \ \ \ \ \ \ uals (except $\bm{e}_{t}$) to get $N_{pa}$ parent chromosomes;
		\State Cut off the parent chromosomes at a random place
        \Statex\ \ \ \ \ \ \ \ \ ${c}$ for $c=1,\cdots,K+N$;
		\State Hybridize them to produce children chromosomes;
		\State Chose $N_{mu}$ chromosome and mutated places $n_1',$
        \Statex\ \ \ \ \ \ \ \ $n_2'$ at random for $n_1'\!=\!1,\!\cdots\!,K$ and $n_2'\!=\!1,\!\cdots\!,N$;
		\State Update $p_{n_1'}$ with a random value subject to \eqref{s1};
		\State Update $\theta_{n_2'}$ with a random value subject to \eqref{s2};
	    \EndIf
		\EndFor
	\end{algorithmic}
\end{algorithm}
\section{Numerical Results}\vspace{-0.1cm}
In this section, we provide simulation results to demonstrate the correctness of our derivations and evaluate the performance of the proposed algorithm. The RIS is deployed at (30 m, 0 m, 8 m) with $N=32$ elements. The number of D2D user-pairs are set to $K=6$, which are distributed in a rectangular plane with four points (0 m, 0 m, 1.6 m), (60 m, 0 m, 1.6 m), (0 m, 25 m, 1.6 m) and (60 m, 25 m, 1.6 m). The large-scale fading coefficients are $PL = -30 - 10\chi {\log _{10}}\left( {d} \right)$ dB, where $\chi$ is the path-loss exponent, and $d$ is the distance between the D2D users (or the D2D user and the RIS) in meters. We set $\chi_{d}=3.8$ for the direct links and $\chi_r=2.2$ for the reflecting links. The AoA and AoD of all channels are randomly generated from $[ {0,2\pi } )$. The Rician factors are set to 10 dB. Unless stated otherwise, we set $P=30$ dBm, $B=3$, $\kappa_{\tilde \theta}=4$, $\sigma_F^2=-70$ dBm, $\sigma_j^2=-80$ dBm \cite{7}, $P_{DC}=-5$ dBm, $P_{SW}=-10$ dBm, $\varepsilon =0.8$ \cite{15}. Note that the curves marked ``PCPSO" represent joint power control and phase shifts optimization while the curves marked ``PSO" represent optimizing the phase shifts with equal power allocation.
\begin{figure}\vspace{-1cm}
  \centering
  \includegraphics[scale=0.5]{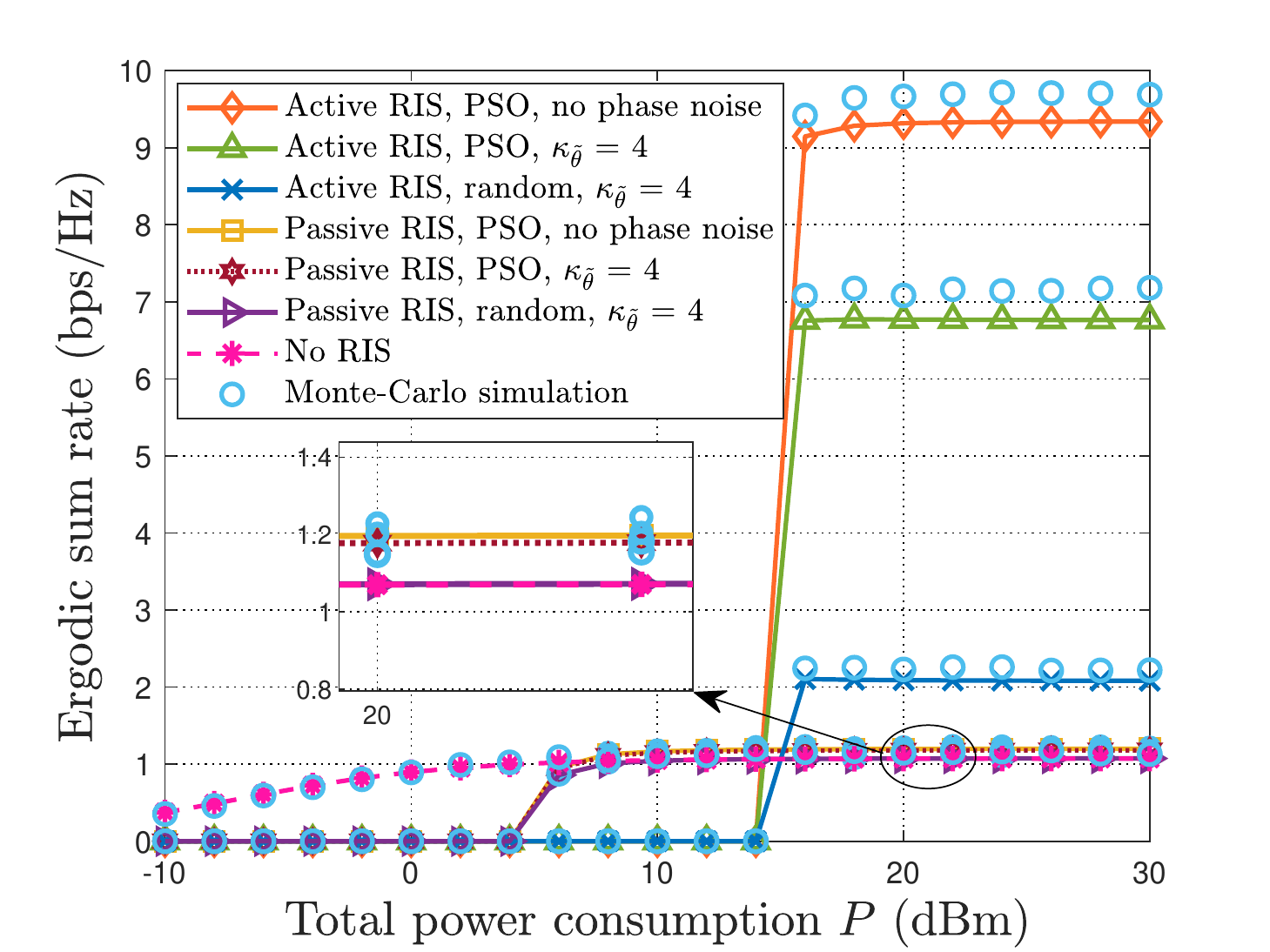}\vspace{-0.3cm}
  \caption{The ergodic sum rate versus the total power consumption.}\label{transmit_power}\vspace{-0.5cm}
\end{figure}
\begin{figure}
\centering
\includegraphics[scale=0.5]{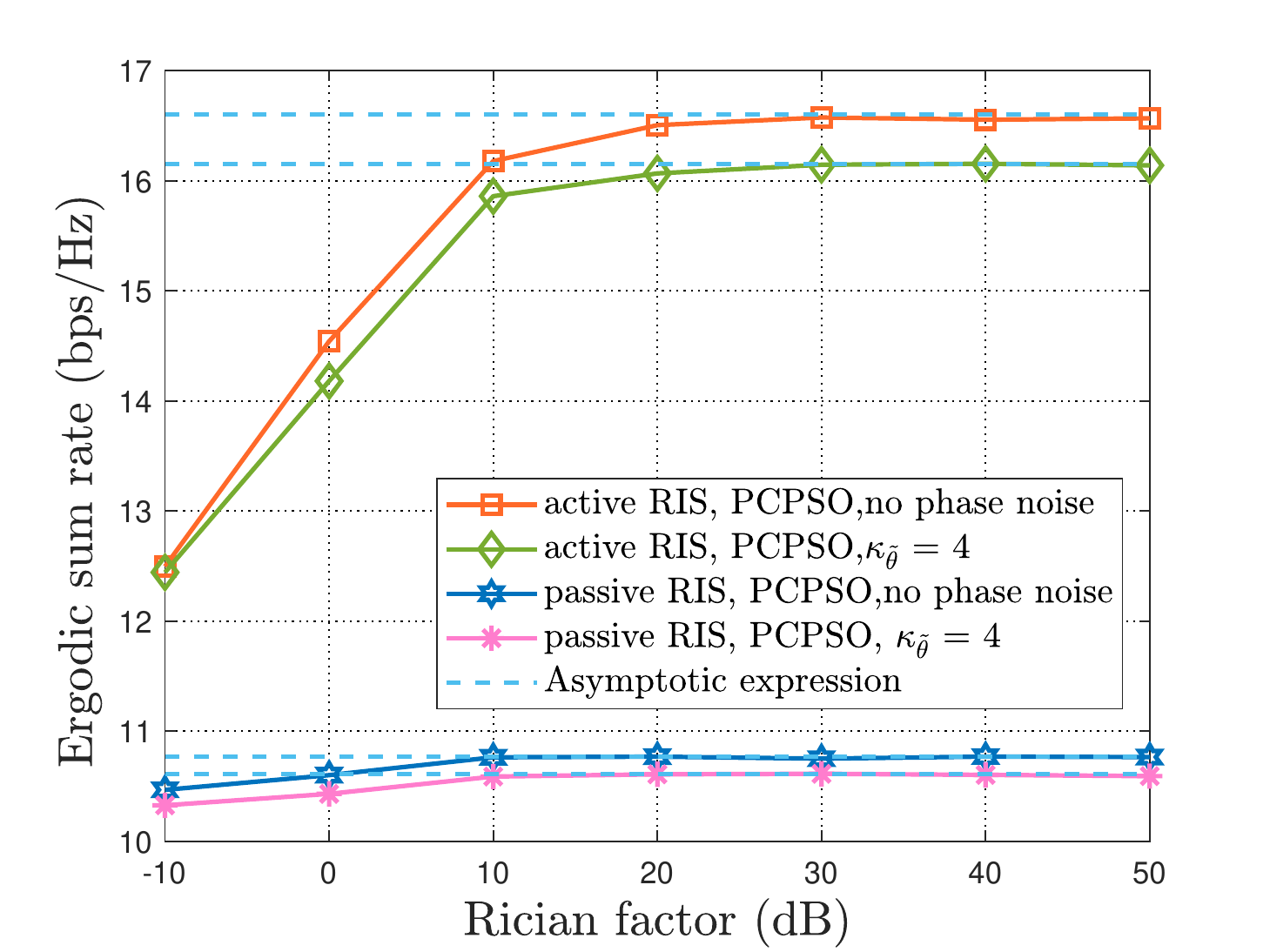}\vspace{-0.3cm}
\caption{The ergodic sum rate versus the Rician factor.}\label{Rician}\vspace{-0.6cm}
\end{figure}

Fig. \ref{transmit_power} depicts the ergodic sum rate versus $P$ with both random and optimized phase shifts in different cases. When $P$ is too small to support hardware power consumption, the ergodic rate is set to zero, i.e., $\tilde{R}_j=0$ when $P\le N(P_{DC}+P_{SW})$ and $\tilde{R}_j^{pas}=0$ when $P\le NP_{SW}$. It can be observed that the closed-form expression in \eqref{R_appro_act} matches well with the Monte-Carlo simulation in \eqref{monte}, which validates the accuracy of our derived results. We observe that the active RIS outperforms the passive RIS while the passive RIS achieves negligible performance gain compared with the case without RIS. As expected, the ergodic sum rates without phase noise are higher. In addition, the proposed GA-based method shows its effectiveness on rate improvement compared with the random phase shifts, which demonstrates the importance of designing the phase shifts. It also indicates that the active RIS can overcome the ``multiplicative fading" effect and achieve noticeable performance gains compared with the passive RIS.\vspace{-0.4cm}
\begin{figure}
\centering\vspace{-1cm}
\includegraphics[scale=0.5]{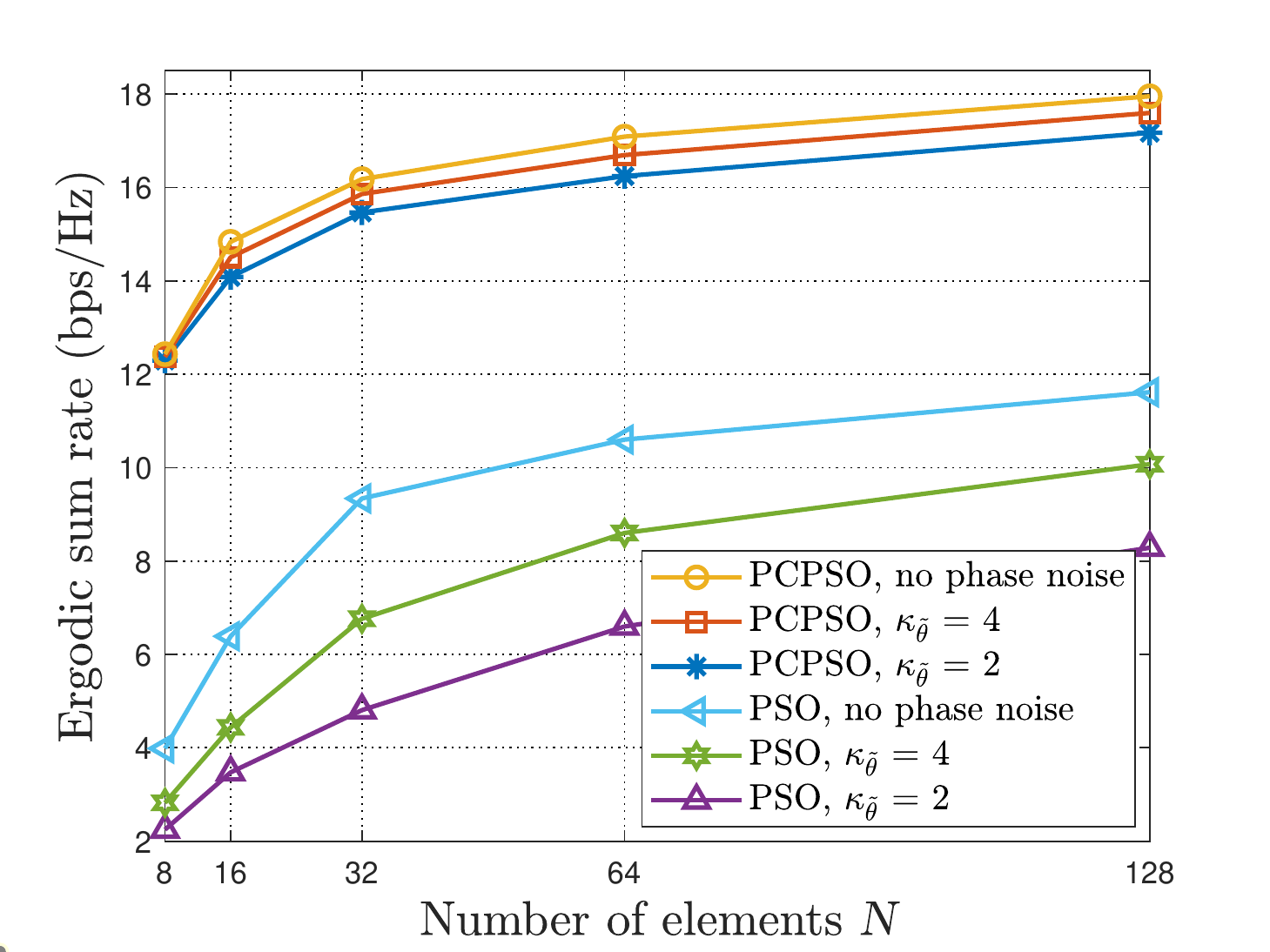}\vspace{-0.3cm}
\caption{The ergodic sum rate versus the number of elements $N$.}\label{N}\vspace{-0.6cm}
\end{figure}

In Fig. \ref{Rician}, the curves marked as ``Asymptotic expression" are obtained by Corollary \ref{Rician_infinity}. The figure shows that the ergodic sum rates notably increase as Rician factors increase at the beginning and converge to constants when the LoS components are large. In other words, the ergodic sum rates are relatively higher in poor scattering environment, which is consistent with the conclusion in \cite{16}. Hence, the RIS should be deployed in the places with poor scatters.

It can be seen from Fig. \ref{N} that the ergodic sum rates increase with the number of reflecting elements $N$ in the active RIS aided system. Even a few elements may improve the ergodic sum rates significantly. In addition, the ergodic sum rates decrease with the decrease of the phase noise parameter $\kappa_{\tilde \theta}$. Increasing the number of reflecting elements achieves limited performance gain when $N$ is large. However, the power control can achieve potential gains with the increase of $N$, and is shown to be an effective way on performance improvement compared with equal power allocation.
Therefore, deploying an active RIS with a small number of reflecting elements is sufficient to achieve significant performance gains.
\vspace{-0.4cm}
\section{Conclusion}
This letter investigated the system performance for a multi-pair D2D communication system aided by an active RIS over spatially correlated channels with phase noise and direct link. The approximate closed-form expression of the ergodic sum rate was derived and analyzed under different channel conditions. The impact of the phase noise on the ergodic sum rate was also discussed. The proposed GA-based method was effective for ergodic sum rate maximization. The simulation results showed that the active RIS outperforms the passive RIS.
\vspace{-0.2cm}
\begin{appendices}
\vspace{-0.4cm}
\section{}\label{Theorem1}
By using \cite[\emph{Lemma 1}]{17}, the ergodic rate can be approximated by
\begin{align}\label{appro}
R_j \!\approx\! {\log _2}\!\!\left(\!\!1\!+\!\!\frac{{{P_{j}}{\mathbb{E}}\!\left\{{{\left| g_{j,j} \right|}^2}\right\}}}{{
\sum\nolimits_{i \ne j}^K\!{{P_{i}}{\mathbb{E}}\!\left\{\!{{\left| g_{i,j} \right|}^2}\!\right\}}
\!\!+\! {\mathbb{E}}\!\left\{\!{{\left\| {{\bf{g}}_{B,j}^T\bm{\Lambda \Theta\Phi }}\right\|}^2}\!\right\}\!\sigma_F^2 \!+\! \sigma _j^2}}\!\!\right)\!\!.
\end{align}

Since $\tilde {\bf{g}}_{A,i}$, $\tilde {\bf{g}}_{B,i}$ and $h_{i,j}$ are independent of each other,
we have ${\mathbb{E}}\left\{ {{{\left| g_{i,j} \right|}^2}} \right\}={\mathbb{E}}\left\{ {{{\left| {{\bf{g}}_{B,j}^T{\bm{\Lambda\Theta\Phi }} {{{\bf{g}}}_{A,i}}} \right|}^2}} \right\}+\sigma_{i,j}^2+2{\mathbb{E}}\left\{{{\mathop{\rm Re}\nolimits} \left\{ {{\bf{g}}_{B,j}^T{\bf{\Lambda \Theta \Phi }}{{\bf{g}}_{A,i}}h_{i,j}^*} \right\}}\right\}$.
First, we derive that
\begin{align}\label{E1}
{\mathbb{E}}\{& {{{| {{\bf{g}}_{B,j}^T{\bm{\Lambda\Theta\Phi }} {{{\bf{g}}}_{A,i}}}|}^2}} \}
\!\!=\!{\mathbb{E}}\!\left\{\! {{\bf{g}}_{B,j}^T\!{\bm{\Lambda\Theta\Phi }} {{{\bf{g}}}_{A,i}}{\bf{g}}_{A,i}^H\!{{\bm{\Phi}}^H}\!{\bm \Theta ^H}\!{{\bf{\Lambda }}^H}\!{\bf{g}}_{B,j}^*}\!\right\}\nonumber\\
 \buildrel (a) \over =\  &\tau_{i,j}(\gamma _{A,i}\gamma _{B,j}{\mathbb{E}}\left\{ {{\bf{\bar g}}_{B,j}^T{\bm{\Lambda\Theta\Phi }} {{{\bf{\bar g}}}_{A,i}}{\bf{\bar g}}_{A,i}^H{{\bm{\Phi}}^H}{{\bf{\Theta }}^H}{{\bm{\Lambda }}^H}{\bf{\bar g}}_{B,j}^*} \right\}\nonumber\\
&+ \gamma _{B,j}{\mathbb{E}}\left\{ {{\bf{\bar g}}_{B,j}^T{\bm{\Lambda\Theta\Phi }} {{{\bf{\tilde g}}}_{A,i}}{\bf{\tilde g}}_{A,i}^H{{\bm{\Phi}}^H}{{\bf{\Theta }}^H}{{\bm{\Lambda }}^H}{\bf{\bar g}}_{B,j}^*} \right\}\nonumber\\
&+ \gamma _{A,i}{\mathbb{E}}\left\{ {{\bf{\tilde g}}_{B,j}^T{\bm{\Lambda\Theta\Phi }} {{{\bf{\bar g}}}_{A,i}}{\bf{\bar g}}_{A,i}^H{{\bm{\Phi}}^H}{{\bf{\Theta }}^H}{{\bm{\Lambda }}^H}{\bf{\tilde g}}_{B,j}^*} \right\}\nonumber\\
&+ {\mathbb{E}}\left\{ {{\bf{\tilde g}}_{B,j}^T{\bm{\Lambda\Theta\Phi }} {{{\bf{\tilde g}}}_{A,i}}{\bf{\tilde g}}_{A,i}^H{{\bm{\Phi}}^H}{{\bf{\Theta }}^H}{{\bm{\Lambda }}^H}{\bf{\tilde g}}_{B,j}^*} \right\}),
\end{align}
where $(a)$ is derived by removing zero terms. Then, we can obtain the following result according to the Euler's formula.
\begin{align}\label{1}
\!\!\!\!&{\mathbb{E}}\!\left\{ {{\bf{\bar g}}_{B,j}^T{\bm{\Lambda\Theta\Phi }} {{{\bf{\bar g}}}_{A,i}}{\bf{\bar g}}_{A,i}^H{{\bf{\Phi}}^H}{{\bf{\Theta }}^H}{{\bm{\Lambda }}^H}{\bf{\bar g}}_{B,j}^*} \right\}
\nonumber\\
=&\eta^2{\mathbb{E}}\!\left\{\!\! {{{\left| {\sum\nolimits_{n = 1}^N \!\!\!\!\!{{e^{j\left( {{\theta _n} + {{\tilde \theta }_n} } \right)\!\!}}\left[{{{\bf{\bar g}}}_{B,j}}\right]_n\!\left[{{{\bf{\bar g}}}_{A,i}}\right]_n}\! } \right|}^2}}\! \right\}\! \!=\! \eta^2(\!N\!+\!\Gamma_{i,j}\!),
\end{align}\vspace{-0.1cm}
where
\begin{align}
\label{Gamma}&{\Gamma _{i,j}} = 2{\kappa ^2}\!\!\!\!\!\!\!\!\sum\limits_{1 \le q <p \le N} \!\!\!\!\!\!\!\!{\cos \left( {{\theta _p} - {\theta _q} + (2\pi/ \lambda) \left( { \bar{s}_{p,q}^{i,j}+ \bar{t}_{p,q}^{i,j}} \right)} \right)},
\end{align}
$\bar{s}_{p,q}^{i,j}=\left( {h\left( p \right) - h\left( q \right)} \right){d_H}\left( {\sin \phi _i^a\cos \phi _i^e + \sin \zeta _j^a\cos \zeta _j^e} \right)$ and $\bar{t}_{p,q}^{i,j}=\left( {v\left( p \right) - v\left( q \right)} \right){d_V}\left( {\sin \phi _i^e + \sin \zeta _j^e} \right)$. Here, $\kappa=\frac{I_1(\kappa_{\tilde \theta})}{I_0(\kappa_{\tilde \theta})}$, and $I_p(\kappa_{\tilde \theta})$ is the modified Bessel function of the first kind and order $p$ \cite{13}.
Next, applying ${\mathbb{E}}\left\{ {{\bf{\Phi R}}{{\bf{\Phi }}^H}} \right\}={\mathbb{E}}\left\{ {{{\bf{\Phi }}^H}{\bf{R\Phi }}} \right\}={\kappa ^2}{\bf{R}}+\left( {1 - {\kappa ^2}} \right){{\bf{I}}_{N\times N}}$, we can derive that
\begin{align}\label{2}
&{\mathbb{E}}\left\{ {{\bf{\bar g}}_{B,j}^T{\bm{\Lambda\Theta\Phi }} {{{\bf{\tilde g}}}_{A,i}}{\bf{\tilde g}}_{A,i}^H{{\bm{\Phi }}^H}{{\bf{\Theta }}^H}{{\bm{\Lambda }}^H}{\bf{\bar g}}_{B,j}^*} \right\}\nonumber\\
=& {\eta ^2}{\mathbb{E}}\left\{ {{\bf{\bar g}}_{B,j}^T{\bf{\Theta \Phi }}{\mathbb{E}}\left\{ {{{{\bf{\tilde g}}}_{A,i}}{\bf{\tilde g}}_{A,i}^H} \right\}{{\bf{\Phi }}^H}{{\bf{\Theta }}^H}{\bf{\bar g}}_{B,j}^*} \right\}\nonumber\\
=& {\eta ^2}{\mathbb{E}}\left\{ {{\bf{\bar g}}_{B,j}^T{\bf{\Theta }}{\mathbb{E}}\left\{ {{\bf{\Phi R}}{{\bf{\Phi }}^H}} \right\}{{\bf{\Theta }}^H}{\bf{\bar g}}_{B,j}^*} \right\}\nonumber\\
=& {\eta ^2}\left( {{\kappa ^2}{\mathbb{E}}\left\{ {{\bf{\bar g}}_{B,j}^T{\bf{\Theta R}}{{\bf{\Theta }}^H}{\bf{\bar g}}_{B,j}^*} \right\} + \left( {1 - {\kappa ^2}} \right){\mathbb{E}}\left\{ {{\bf{\bar g}}_{B,j}^T{\bf{\bar g}}_{B,j}^*} \right\}} \right)\nonumber\\
=&{\eta ^2} ({{\kappa ^2}N+L_{{\zeta _j^a},{\zeta _j^e}}  + ( {1 - {\kappa ^2}} )N} )={\eta ^2}(N+L_{{\zeta _j^a},{\zeta _j^e}}),
\end{align}where
\begin{align}\label{L}
L_{x,y}\!\!=2{\kappa ^2}\!\!\!\!\!\!\!\!\sum\limits_{1 \le q < p \le N}\!\!\!\!\!\!\!\!{{r_{p,q}}\cos\left( {{\theta _p}\!-\!{\theta _q}\!+\!(2\pi/ \lambda)\left( { s_{p,q}^{x,y}+ t_{p,q}^{x,y}} \right)} \right)}.
\end{align}
$r_{p,q}$ is given by \eqref{R}, $s_{p,q}^{x,y}=\left( {h\left( p \right) - h\left( q \right)} \right){d_H}\sin x\cos y$ and $t_{p,q}^{x,y}=\left( {v\left( p \right) - v\left( q \right)} \right){d_V}\sin y$. Similarly, we have
\begin{align}\label{3}
&{\mathbb{E}}\!\left\{\! {{\bf{\tilde g}}_{B,j}^T{\bm{\Lambda\Theta\Phi }} {{{\bf{\bar g}}}_{A,i}}{\bf{\bar g}}_{A,i}^H{{\bm{\Phi }}^H}{{\bf{\Theta }}^H}{{\bm{\Lambda }}^H}{\bf{\tilde g}}_{B,j}^*} \!\right\}\nonumber\\
=&{\eta ^2}{\mathbb{E}}\left\{ {{\bf{\bar g}}_{A,i}^H{{\bf{\Phi }}^H}{{\bf{\Theta }}^H}{\mathbb{E}}\left\{ {{\bf{\tilde g}}_{B,j}^{\rm{*}}{\bf{\tilde g}}_{B,j}^T} \right\}{\bf{\Theta \Phi }}{{{\bf{\bar g}}}_{A,i}}} \right\}\nonumber\\
=&{\eta ^2}\left( {{\kappa ^2}{\mathbb{E}}\left\{ {{\bf{\bar g}}_{A,i}^H{\bf{\Theta R}}{{\bf{\Theta }}^H}{\bf{\bar g}}_{A,i}} \right\} + \left( {1 - {\kappa ^2}} \right){\mathbb{E}}\left\{ {{\bf{\bar g}}_{A,i}^H{\bf{\bar g}}_{A,i}} \right\}} \right)\nonumber\\
=&{\eta ^2} ({{\kappa ^2}N+L_{{\phi _i^a},{\phi _i^e}}  + ( {1 - {\kappa ^2}} )N} )={\eta ^2}(N+L_{{\phi _i^a},{\phi _i^e}}),\\
\label{4}
&{\mathbb{E}}\{ {{\bf{\tilde g}}_{B,j}^T{\bm{\Lambda\Theta\Phi }} {{{\bf{\tilde g}}}_{A,i}}{\bf{\tilde g}}_{A,i}^H{{\bm{\Phi }}^H}{{\bf{\Theta }}^H}{{\bm{\Lambda }}^H}{\bf{\tilde g}}_{B,j}^*}\}\nonumber\\
= &\ {\eta ^2}(N+2{\kappa ^2}\!\!\!\!\!\!\!\!\sum\limits_{1 \le q<p \le N}\!\!\!\!\!\!\!\!{{r_{p,q}^2}\cos\left( {{\theta _p}-{\theta _q}}\right)})\buildrel \Delta \over={\eta ^2}(N+L_0).
\end{align}
Substituting \eqref{1}, \eqref{2}, \eqref{3} and \eqref{4} into \eqref{E1}, we arrive at
\begin{align}\label{E1_result}
{\mathbb{E}}\!&\left\{\!{{\left|{{\bf{g}}_{B,j}^T{\bf{\Lambda\Theta\Phi }}{{{\bf{g}}}_{A,i}}}\right|}^2}\!\right\}
\!=\!\eta^2\tau_{i,j}({\gamma _{A,i}}{\gamma _{B,j}}(N\!+\!\Gamma _{i,j})\!+\!\gamma_{A,i}(N\nonumber\\&\!\!+\!L_{{\phi _i^a},{\phi _i^e}})\!+\!\gamma_{B,j}(N\!+\!L_{{\zeta _j^a},{\zeta _j^e}})\!+\!(N\!+\!L_0))\buildrel \Delta \over=\eta^2\Omega _{i,j}.
\end{align}
Similarly, we have \vspace{-0.1cm}
\begin{equation}\label{noise_result}
{\mathbb{E}}\left\{ {{{\left\| {{\bf{g}}_{B,j}^T\bm{\Lambda \Theta\Phi } } \right\|}^2}} \right\}
\!\!=\eta^2N\beta_j.\vspace{-0.1cm}
\end{equation}
Note that ${\mathbb{E}}\{e^{j\tilde \theta_n}\}=\kappa$, we can derive that
\begin{align}\label{Upsilon}
&{\mathbb{E}}\left\{{\mathop{\rm Re}\nolimits}\{{{\bf{g}}_{B,j}^T{\bf{\Lambda \Theta \Phi }}{{\bf{g}}_{A,i}}h_{i,j}^*}\}\right\}
=\!{\mathop{\rm Re}\nolimits}\{ {{\bf{\bar g}}_{B,j}^T{\bf{\Lambda \Theta {\mathbb{E}}\{\Phi }\}}{{\bf{\bar g}}_{A,i}}\bar h_{i,j}} \}\nonumber\\
=&\eta c_{i,j} \sum\nolimits_{n = 1}^N \!\!\!\!\!\!\!\cos\left({{\theta _n} +{2\pi l_{i,j}^n/\lambda } } \right)\buildrel \Delta \over=\eta c_{i,j}{\Upsilon _{i,j}},
\end{align}
where $c_{i,j}\!=\!\kappa\sigma_{i,j}\sqrt{{{{\tau _{i,j}}{\gamma _{A,i}}{\gamma _{B,j}}{\gamma _{i,j}}}}}/\!\!\sqrt{1 \!+\! {\gamma _{i,j}}}$ and $l_{i,j}^n\!=\!h\left( n \right)$ $ {d_H}( {\sin \phi _i^a\cos \phi _i^e \!+\! \sin \zeta _j^a\cos \zeta _j^e})\!+\!v\left( n \right){d_V}\left( {\sin \phi _i^e\! +\! \sin \zeta _j^e} \right)$.

We can substitute \eqref{eta}, \eqref{E1_result} - \eqref{Upsilon} into \eqref{appro} to complete the proof.
\section{}\label{N_infty}\vspace{-0.1cm}
From \eqref{1}, we can infer that
\begin{equation}
\label{Gamma_max}
 N \!+\! {\tilde \Gamma _{i,i}}\!\le\! {\left( {\sum\nolimits_{n = 1}^N {\left| {{e^{j{\theta _n}}}{{\left[ {{{{\bf{\bar g}}}_{B,i}}} \right]}_n}{{\left[ {{{{\bf{\bar g}}}_{A,i}}} \right]}_n}} \right|} } \right)^2} = {N^2},
\end{equation}
which indicates that the optimal phase shift of each element on the active RIS with ideal hardware is $\theta _n^{opt} =  - \frac{{2\pi }}{\lambda }( {h\left( n \right){d_H}\left( {\sin \phi _i^a{\rm{cos}}\phi _i^e + \sin \zeta _i^a\cos \zeta _i^e} \right) + v\left( n \right){d_V}( {\sin \phi _i^e}}$ ${{ + \sin \zeta _i^e} )} ) + C$, where $C$ is an arbitrary constant. Recalling ${\Omega_{i,i}^*}$ in Remark \ref{onepair}, we have ${\Omega_{i,i}^{
*,opt}}\!=\!\tau_{i,i}(\gamma_{A,i}\gamma_{B,i}N^2+(\gamma_{A,i}+$ $\gamma_{B,i})N+N)$. Substituting ${\Omega _{i,i}^*}$ with ${\Omega _{i,i}^{*,opt}}$ into ${\tilde R_1}$ yields
${{\tilde R_1^{opt}}} $ $\!\!\!\!\!={\log _2}\left({1+ \frac{{{E_u}{P_R}{\tau _{i,i}}( {{\gamma _{A,i}}{\gamma _{B,i}}+  \frac{1}{N}({\gamma _{A,i}}+  {\gamma _{B,i}}) +  \frac{1}{N}} )}}{{{P_R}{\beta _i}\sigma _F^2+  {\frac{E_u}{N}{\alpha _i}\sigma _i^2+  \sigma _F^2} \sigma _i^2}}} \right)$. When $N\to\infty$, we have $\frac{1}{N}\to 0$ and $R_1\to{{\tilde R_1^{opt}}} $. We arrive at the final result after removing the zero terms.\vspace{-0.28cm}
\end{appendices}
\bibliographystyle{IEEEtran}
\bibliography{IEEEabrv,Refer}
\end{document}